\documentclass[10pt, twocolumn, a4paper]{IEEEtran}
\usepackage{cite}
\ifCLASSINFOpdf
   \usepackage[pdftex]{graphicx}
  \DeclareGraphicsExtensions{.pdf,.jpeg,.png}
\else
  \usepackage[dvips]{graphicx}
  \DeclareGraphicsExtensions{.eps}
\fi

 \usepackage{nopageno}
\usepackage{relsize}
\usepackage[cmex10]{amsmath}
\usepackage{algorithmic}
\usepackage{array}

\usepackage{dblfloatfix}
\ifCLASSOPTIONcaptionsoff
  \usepackage[nomarkers]{endfloat}
 \let\MYoriglatexcaption\caption
 \renewcommand{\caption}[2][\relax]{\MYoriglatexcaption[#2]{#2}}
\fi

\usepackage{url}
\usepackage{epsfig}
\usepackage{multicol}
\usepackage{verbatim}
\usepackage{amssymb}
\usepackage{epstopdf}
\usepackage{amsthm}
\usepackage{soul}
\usepackage{amsfonts}%
\usepackage{enumerate}%
\usepackage{psfrag}
\usepackage{setspace}
\usepackage{dsfont}
\usepackage{balance}

\usepackage{algorithm}
\usepackage{float}

\usepackage{mdwmath}
\usepackage{mdwtab}
\usepackage[section]{placeins}
\usepackage{caption}
\usepackage{subcaption}
\usepackage{xcolor}
\usepackage{mathrsfs}

\newtheorem{theorem}{Theorem}
\newtheorem{proposition}{Proposition}
\newtheorem{remark}{Remark}
\newtheorem{definition}{Definition}
\newtheorem{lemma}{Lemma}

\renewcommand{\vec}[1]{\mathbf{#1}}
\begin{document}

\title{Secure Degrees of Freedom of the MIMO Multiple Access Channel with Multiple unknown Eavesdroppers}


\author{\IEEEauthorblockN{Mohamed Amir, Tamer Khattab, Tarek Elfouly, Amr Mohamed}
\IEEEauthorblockA{Qatar University\\
Email: mohamed.amir@qu.edu.qa, tkhattab@ieee.org, tarekfouly@qu.edu.qa, amrm@ieee.org \vspace{-5mm}}
\thanks{This research was made possible by NPRP 5-559-2-227 grant from the Qatar National Research Fund (a member of The Qatar Foundation). The statements made herein are solely the responsibility of the authors.}}

\maketitle
\thispagestyle{empty}
\pagestyle{empty}

\begin{abstract}

We investigate the secure degrees of freedom (SDoF) of a two-transmitter Gaussian multiple access  channel with multiple antennas at the transmitters, the legitimate receiver with the existence of an unknown number of eavesdroppers each with a number of antennas less than or equal to a known value $N_E$. The channel matrices between the legitimate transmitters and the receiver are available everywhere, while the legitimate pair does not know the eavesdroppers' channels matrices. We provide the exact sum SDoF for the considered system. A new comprehensive upperbound is deduced and a new achievable scheme based on utilizing jamming is exploited. We prove that Cooperative Jamming is SDoF optimal even without the instantaneous eavesdropper CSI available at the transmitters.   

\end{abstract}

\section{Introduction}
The noisy wiretap channel was first studied by Wyner \cite{wyner}, in which a
legitimate transmitter (Alice) wishes to send a message to a legitimate receiver (Bob), and hide it from an eavesdropper (Eve). Wyner proved that Alice can send positive secure rate using channel coding. He derived capacity-equivocation region for the degraded
wiretap channel. Later, Csiszar and Korner found capacity-equivocation the region for the
general wiretap channel \cite{csi}, which was
extended to the Gaussian wiretap channel by Leung-Yan-
Cheong and Hellman \cite{leu}.

A significant amount of work was carried thereafter to study the information theoretic physical layer security for different network models. The relay assisted wiretap channel was studied in \cite{secop}. The secure degrees of freedom (SDoF) region of multiple access channel (MAC) was presented in \cite{sennur_mac}. The SDoF is the the pre-log of the secrecy capacity region in the high-SNR regime. 
Using MIMO systems for securing the message was an intuitive extension due to the spatial gain provided by multiple antennas.  The MIMO wiretap channel was studied in \cite{mimo_wire,mimo_secure,mimo_note,mimo_hass,yener_coop,mimo_shafie,mimo_confidential} and the secrecy capacity was identified in \cite{mimo_hass}.
  
Meanwhile, the idea of cooperative jamming was proposed in \cite{yener_coop},
 where some of the users transmit independent and identically distributed (i.i.d.)
Gaussian noise towards the eavesdropper to improve the sum secrecy rate. Cooperative jamming was used for deriving the SDoF for different networks. In \cite{sennur_mac}, cooperative jamming was used to jam the eavesdropper and proved that the K-user MAC with single antenna nodes can achieve $\frac{K(K-1)}{K(K-1)+1}$ SDoF.

In this paper, we study the MIMO MAC channel with
unknown number of eavesdroppers where the eavesdroppers'
channel coefficients are not available at the legitimate transmitters
and receiver, but only second order statistics of the
eavesdropper channel is available. The motivation emanates
from the fact that the passive eavesdroppers instantaneous and
exact CSI is hard to obtain because the eavesdropper is passive.
It is essential to determine the Secure DoF under other different
forms of the CSI availability, e.g statistical, alternating or
no CSI. It is easy to prove that if the eavesdropper channel
is completely unknown, then the secrecy capacity is equal to
zero. The secrecy capacity of the MIMO wiretap channel with
known CSI is an upperbound for these channel given equal
number of antennas at both transmitting and receiving sides,
which is in turn upperbounded as
\begin{equation}
I(W_t ; Y ) - I(W_t;Z)
\end{equation}
where $W_t$ is the set of messages to be transmitted, $Y$ is
the legitimate receiver signal and $Z$ is the eavesdropper
signal. So for identifying the secrecy capacity and building a
positive rate secure code for the no CSI case, it must be proven first that
$I(W_t;Z) < I(W_t; Y )$ which is not possible if the channel is
completely unknown.

The major contributions of our work as compared to existing literature can be summarized as follows:
\begin{itemize}
 \item We present the sum SDoF of the multiple access channel with unknown fading eavesdroppers channels.
\item We present the sum SDoF of the multiple access channel with known eavesdropper channels with constant or time varying channels, closing an open problem since the best  known achievable region was presented in \cite{khan}. we show that it has the same sum SDoF as the previous unknown fading channels case.
\item our work incorporates the more general scenario of multiple eavesdroppers.
 \item We study the more comprehensive case where all the eigenvalues of the legitimate channel has non-zero values\footnote{The cases where some of the eigenvalues are equal to zero represent special degraded cases of the more general non-zero eigenvalues case, where the SDoF decreases for every zero eigenvalue till it collapses to the trivial case of zero SDoF for all-zero eigenvalues.}.
 \item we deduce a new upperbound on the SDoF and provide an optimal scheme that achieves the new upperbound.  
\end{itemize}

The paper is organized as follows. Section~\ref{sec:model} defines the system model and the secrecy constraints. The main results are presented in Section~\ref{sec:results}. In Section~\ref{sec:bound}, the new upperbound is derived and the achievable scheme is presented in Section~\ref{sec:scheme}.  The paper is concluded in Section~\ref{sec:conclusion}. We use the following notation, $\vec{a}$ for vectors, $\vec{A}$ for matrices, $\vec{A}^{\dagger}$ for the hermitian transpose of $\vec{A}$, $[A]^+$ for the $\max{A,0}$ and $\mathrm{Null}(\vec{A})$ to define the nullspace of $\vec{A}$.

%%%%%%%%%%%%%%%%%%%%%%%%%%%%%%%

\section{System model}
\label{sec:model}

We consider a communication system of two transmitters and a single receiver in vicinity of an unknown number of  passive eavesdroppers. Transmitters one and two are equipped with $M_1$ and $M_2$ antennas, respectively.  The legitimate receiver is equipped with $N$ antennas, while the $j$th eavesdropper is equipped with $N_{Ej}\leq N_E$ antennas, where $N_E$ is a constant known to the eavesdroppers. Let $\vec{x}_i$ denote the $M_i \times 1$ vector of symbols to be transmitted by transmitter $i$, where $i\in\{1,2\}$. We can
write the received signal at the legitimate receiver at time (sample) $k$ as
\begin{equation}\label{Received_signal}
\vec{Y}(k)=\sum_{i=1}^2\vec{H}_i \vec{V}_i\vec{x}_i(k)+ \vec{n}(k)
\end{equation}
and the received signal at the $j$th eavesdropper as
\begin{equation}\label{Received_signal}
\vec{Z}_{j}(k)=\sum_{i=1}^2\vec{G}_{i,j}(k) \vec{V}_i\vec{x}_i(k)+ \vec{n}_{Ej}(k),
\end{equation}

\noindent where $\vec{H}_i$ is the $N \times M_i$ matrix containing
the channel coefficients from transmitter $i$ to the receiver,  $\vec{G}_{i,j}(k)$ is the $N_{Ej} \times M_i$ matrix containing the i.i.d time varying channel coefficients from transmitter $i$ to the eavesdropper $j$ drawn from a continuous distribution with mean $\eta$ and variance $\sigma_e^2$,  $\vec{V}_i$ is the precoding unitary matrix (i.e. $\vec{V}_i\vec{V}_i^\dagger = \vec{I}$) at transmitter $i$, $\vec{n}(k)$ and $\vec{n}_{Ej}(k)$ are the $N\times 1$  and the $N_{Ej}\times 1$ additive white Gaussian
noise vectors with zero mean and variance $\sigma^2$ at the legitimate receiver and the $j$th eavesdropper, respectively. We assume that the transmitters do not know any of the eavesdroppers' channels $\vec{G}_{i,j}(k)$. We assume that $N_E< M$, where $M=M_1+M_2$.  

We define the $M_i \times 1$ channel input from legitimate transmitter $i$ as
\begin{equation}
\vec{X}_i(k)= \vec{V}_i \vec{x}_i(k).
\end{equation}

%We denote $E_m$ as the eavesdropper with the maximum number of antennas, $N_E$.
\begin{figure}
  \begin{center}
\hspace{-4mm} \includegraphics[width=.50\textwidth]{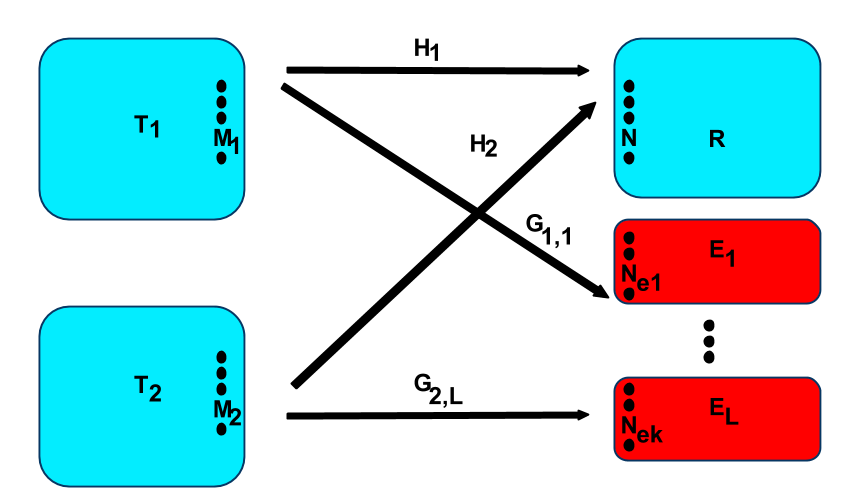}\vspace{-2mm}
\caption{System model}
  \label{sys}
  \end{center}
  \end{figure}

Each transmitter $i$ intends to send a message $W_i$ over $n$ channel uses (samples) to the legitimate receiver simultaneously while preventing the eavesdroppers from decoding its message.  The encoding occurs under  a constrained power given by
\begin{equation}
\text{E}\left\{\vec{X}_i\vec{X}_i^{\dagger}\right\} \leq P \text{  }\forall{i=1,2}
\end{equation}

Expanding the notations over $n$ channel extensions we have $\vec{H}_i^n = \vec{H}_i(1), \vec{H}_i(2), \ldots, \vec{H}_i(n)$, $\vec{G}_{i,j}^{n} = \vec{G}_{i,j}(1), \vec{G}_{i,j}(2), \ldots, \vec{G}_{i,j}(n)$ and similarly the time extended channel input, $\vec{X}_i^n$, time extended channel output at legitimate receiver,  $\vec{Y}^n$ and time extended channel output at eavesdropper $j$, $\vec{Z}_{j}^n$ as well as noise at legitimate receiver, $\vec{n}^n$ and noise at eavesdroppers, $\vec{n}_{Ej}^n$.

At each transmitter, the message $W_i$ is uniformly and independently chosen from a set of possible secret messages for transmitter $i$, $\mathcal{W}_i = \{1,2, \ldots, 2^{nR_i}\}$.  The rate for message $W_i$ is $R_i \triangleq \frac{1}{n} \log\left|\mathcal{W}_i\right|$, where $|\cdot|$ denotes the cardinality of the set. Transmitter $i$ uses a stochastic encoding function $f_i: W_i \longrightarrow \vec{X}_i^n$ to map the secret message into a transmitted symbol.  The receiver has a decoding function $\phi: \vec{Y}^n \longrightarrow (\hat{W}_1,\hat{W}_2)$, where $\hat{W}_i$ is an estimate of $W_i$.
\begin{definition}
A secure rate tuple $(R_1, R_2)$ is said to be achievable if for any $\epsilon > 0$ there exist $n$-length codes such that the legitimate receiver can decode the messages reliably, i.e.,
\begin{equation}
\text{Pr}\{(W_1,W_2) \neq (\hat{W}_1, \hat{W}_2)\} \leq \epsilon
\end{equation}
and the messages are kept information-theoretically secure against the eavesdroppers, i.e.,
\begin{equation}\label{eqn:cond}
\lim\limits_{n\longrightarrow\infty} \frac{1}{n}H(W_1, W_2|\vec{Z}_{j}^{n})\geq \lim\limits_{n\longrightarrow\infty} \frac{1}{n}H(W_1, W_2)-\epsilon \\
\end{equation}

\noindent where $H(\cdot)$ is the Entropy function and~\eqref{eqn:cond} implies the secrecy for any subset $\mathbb{S} \subset \{1,2\}$ of messages including individual messages~\cite{sennur_mac}.
\begin{equation}
\lim\limits_{n\longrightarrow\infty} \frac{1}{n}H(W_i|\vec{Z}_{j}^{n})\geq \lim\limits_{n\longrightarrow\infty} \frac{1}{n}H(W_i) - \epsilon \text{ }\forall  i =1,2
\end{equation}
\end{definition}
\begin{definition}
The sum SDoF is defined as
\begin{equation}
D_s = \lim_{P\rightarrow \infty} \sup{\sum_i \frac{R_i}{\frac{1}{2}\log P}},
\end{equation}
\noindent where the supremum is over all achievable secrecy rate tuples $(R_1, R_2)$, $D_s = d_1 + d_2$, and $d_{1}$ and $d_{2}$ are the secure DoF of transmitters one and two, respectively. 
\end{definition}

\section{Main Results}
\label{sec:results}

\begin{theorem}
The sum SDoF of the two user MAC channel is
\small
\begin{equation}
D_s = \begin{cases}
  M-N_E & \mathscr{C}_1\\
  \frac{1}{2}\max(M_1,N)+\frac{1}{2}\max(M_2,N)-\frac{1}{2}N_E & \mathscr{C}_2\\
  N & \mathscr{C}_3 \\
\end{cases}
\end{equation}
\noindent where the conditions $\mathscr{C}_1$, $\mathscr{C}_2$ and $\mathscr{C}_3$ are given by:
\begin{eqnarray}
\nonumber &&\mathscr{C}_1:   M \leq N\\ 
\nonumber  &&or\;M_1<N, M>N  \text{ and  } N_E \geq 2(M-N)\\
\nonumber  &&or\;M_1>N, M_2<N \text{ and  } N_E \geq [M_1-N+2M_2]\\
\nonumber &&\\
\nonumber &&\mathscr{C}_2:  M_1<N \text{ and } N_E < 2(M-N)\\
\nonumber &&or\; M_1>N, M_2<N \text{ and  } M_1-N\leq N_E< [M_1-N+2M_2]\\
\nonumber &&or\; M_1>N,  M_2\geq N \text{ and  } N_E \geq M-2N  \\
\nonumber &&\\
\nonumber &&\mathscr{C}_3: N_E < [M_1-N]^++[M_2-N]^+
\end{eqnarray}
\end{theorem}

\begin{proof}
To prove the theorem, we deduce the converse in Section~\ref{sec:bound} and provide the achievable scheme in Section~\ref{sec:scheme}.
\end{proof}
\normalsize

%Meanwhile, if transmitter cooperation exists and they can share a common random Jamming signal. The sum secure DoF is
%\begin{equation}
%d_1+d_2 =
%\begin{cases}
%    M-N_E ,& \text{for} N \geq M \\
%    N,              & \text{for} N \leq M
%\end{cases}
%\end{equation}

\section{Converse}\label{sec:bound}
\begin{theorem}
The number of SDoF of the two user MAC channel is upperbound as,
\small
\begin{equation}
D_s\hspace{-1mm}\leq \hspace{-1mm}\min (N,M_1+M_2-N_E,\frac{1}{2}(\max (M_1,N)+\max(M_2,N)-N_E))
\vspace{-2mm}
\end{equation}
\normalfont
\end{theorem}
\begin{proof}
The first bound is the due to limited number of antennas at the receiver which limits the SDoF as
\vspace{-1mm}
\begin{equation}\label{eqn1}
D_s \leq N
\vspace{-1mm}
\end{equation}
The second bound represent the DoF loss caused by the number of eavesdroppers' antennas on the transmitter side, without loss of generality, we provide an upperbound for the case of existence of only one eavesdropper with $N_E$ antennas. The SDoF of the single eavesdropper case is certainly an upperbound for the multiple eavesdroppers case, as increasing the number of eavesdroppers can only reduce the SDoF of the legitimate users. We omit the eavesdropper subscript for simplicity of notation. Suppose that we can added $|M-N|^+$ antennas to the receiver side that wont decrease the SDoF, the sum rate is upperbounded by the rate of an equivalent MIMO wiretap channel with $(M_1+M_2)$ transmit antennas as,
\vspace{-.5mm}
\begin{equation}
R_s \leq \max_{K_x} \log | (\vec{I}+\vec{H}_{11} K_x \vec{H}_{11}^{\dagger})| -\log | (\vec{I}+\vec{G} K_x \vec{G}^{\dagger})
\vspace{-.5mm}
\end{equation}
As $\vec{H}\vec{H}^{\dagger}$ and $\vec{G}\vec{G}^{\dagger}$ are hermitian, they can be diagonalized as $\vec{G}\vec{G}^{\dagger}= \vec{U}_G\vec{\Lambda}_G\vec{U}_G^{\dagger}$, $\vec{H}\vec{H}^{\dagger}= \vec{U}_H\vec{\Lambda}_H\vec{U}_H^{\dagger}$, where $\vec{U}_G\vec{U}_G^{\dagger}=\vec{I}$ and $\vec{U}_G\vec{U}_G^{\dagger}=\vec{I}$. Without loss of generality, Let $\vec{V}= [\vec{V}_L \vec{V}_N]$, where $\vec{V}_N$ contains the $N_E$orthonormal basis of $\vec{G}$, while $\vec{V}_L$ contains the $M-N_E$ basis of the orthogonal complement of $\vec{V}_N$, and  $K_x = \vec{V}\vec{\Lambda}_{K_x}\vec{V}^{\dagger}$. 
\vspace{-.5mm}
\begin{eqnarray}
\small
\label{eq2}\nonumber D_s\hspace{-3mm}&\leq& \hspace{-2.5mm}\lim_{P\rightarrow \infty} \frac{1}{\log P} \big(\max_{K_x} [\log | \vec{I}+\vec{H}_{11}  \vec{V}\vec{\Lambda}_{K_x}\vec{V}^{\dagger} \vec{H}_{11}^{\dagger}| \big)\\
\nonumber &-&\hspace{-2.5mm}\log | \vec{I}+\vec{U}\vec{\Lambda}_G \vec{U}^{\dagger}\vec{V}\vec{\Lambda}_{K_x}\vec{V}^{\dagger}|]\\
\label{eq3}\nonumber&\leq&\hspace{-2.5mm}\lim_{P\rightarrow \infty}\frac{1}{\log P} \big(\max_{\Lambda_{K_x}} \log  |\Lambda_{H}\Lambda_{K_x}| -\log  | \Lambda_G \Lambda_{K_x}| -C_2 \big)\\
\nonumber &\leq&\hspace{-2.5mm}\lim_{P\rightarrow \infty} \frac{1}{\log P}\big(\max_{\Lambda_{K_x}} \prod_{i=1}^{M_1+M_2}{\lambda_{H}^i\lambda_{K_x}^i}- \prod_{i=1}^{N_E}{\lambda_{G}^i\lambda_{K_x}^i\big)}\\
\label{up2}&\leq& \hspace{-2.5mm} M_1+M_2-N_E 
\end{eqnarray}
\normalsize
where $\lambda^i_{K_x}$ is the $i$th diagonal value of $\Lambda_{K_x}$ and similarly defined $\lambda_{G}^i, \lambda_{H}^i$,  (\ref{eq2}) is because $\log |\vec{I}+\vec{A}\vec{B}|= \log |\vec{I}+\vec{B}\vec{A}|$ for the above matrices, (\ref{eq3}) is because $\lim\limits_{P\rightarrow \infty} \frac{\log |\vec{I}+\vec{B}|}{\log P} = \lim\limits_{P\rightarrow \infty} \frac{\log |\vec{B}|}{\log P}  $ for any matrix $\vec{B}$, and because $|\vec{A}\vec{B}|= |\vec{A}||\vec{B}|$ for square matrices, and ${|\vec{V}_{K_x}|, |\vec{V}_H|, |\vec{U}| }$ are independent of $P$. All $C_{i \; : \; i\in\{1,2\}}$ are constants independent of $P$. 
\noindent The third bound represents the DoF loss of each transmitter due to the number of eavesdroppers antennas available.  Let $d_e^1$ and $d_e^2$ be the degrees of freedom for the parts of the messages sent by transmitter one and two, respectively, which can be decoded by the eavesdropper. Note that $\sum_{i=1,2} d_i+d_e^2=M$, because the transmitters is assumed to be using all their $M$ DoF, where from the secure DoF perspective the DoF wasted a non secure message is equivalent to non transmitting over those DoF. Suppose that we can set all coefficients of  $\vec{G}_{1}$ to zero to hide all information sent by transmitter one from the eavesdropper, so the resulting channel becomes identical to the Z channel of Figure~\ref{zch}(b). However,  setting $\vec{G}_{1}$ coefficients to zero cannot decrease the performance of the coding scheme. Therefore, the DoF tuple for the modified channel is upperbounded by
\vspace{-.5mm}
\begin{equation}\label{z1}
d_{1}+d_{2}+d_{e}^2 \leq \text{max}(M_2, N)\\
\end{equation}
\noindent Similarly, using the modified Z channel in Figure~\ref{zch}(a),
\begin{equation}\label{z2}
d_{1}+d_{2}+d_{e}^1\leq \text{max}(M_1,N)\\
\end{equation} 
Moreover, since the eavesdropper has $N_E$ antennas then,
\begin{equation} \label{eve}
d_{e}^1+d_{e}^2=N_E
\end{equation}
where the LHS is strictly equal to the RHS because the transmitter messages is occupying $M>N_E$ dimensions as mentioned above.
Combining (\ref{z1}), (\ref{z2}) and (\ref{eve}) we have,
\begin{equation}\label{c3}
D_s\leq \frac{1}{2}(\max(M_1,N) +\max(M_2, N)-N_E)
\end{equation} 
\begin{figure}  \label{zch}
\vspace{-4mm}
  \begin{center}
\hspace{-4mm} \includegraphics[width=.4\textwidth]{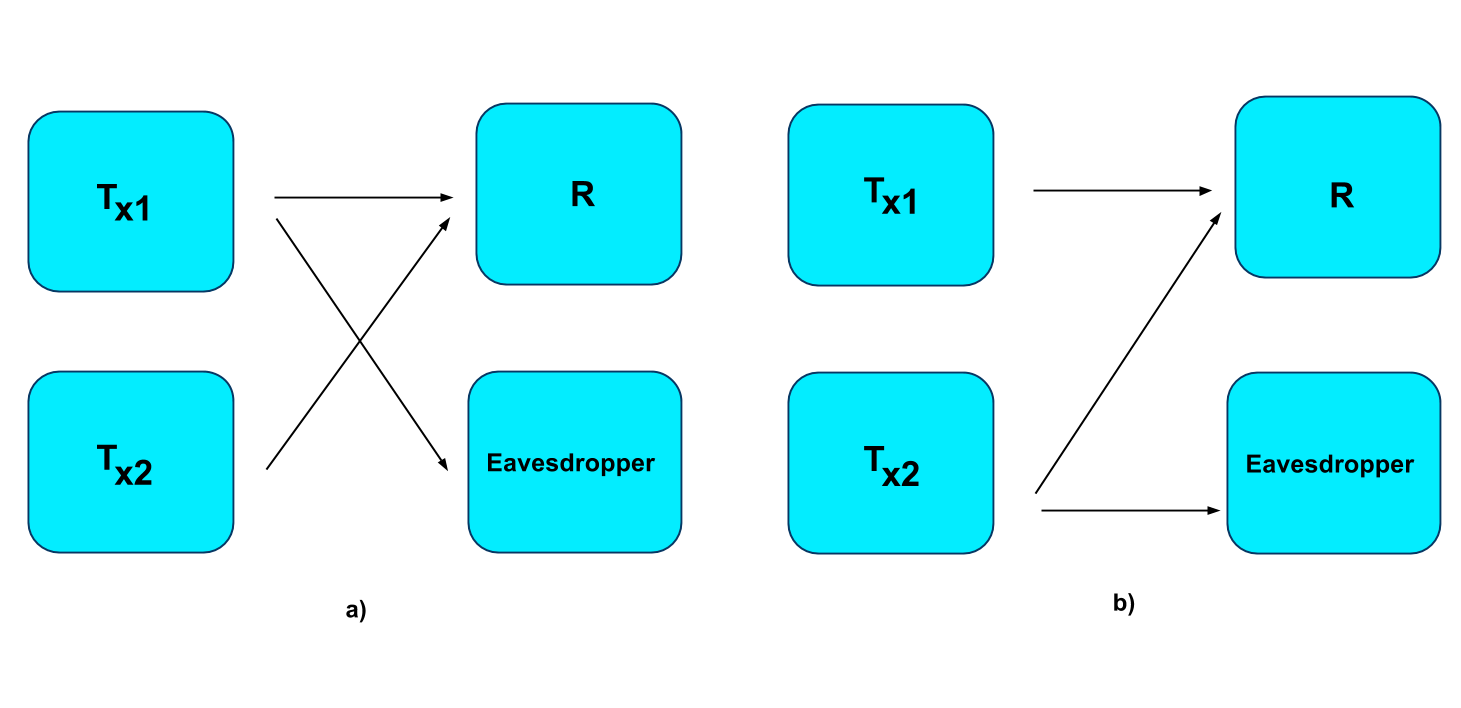}\vspace{-2mm}
\caption{DoF-equivalent Z channel }
\label{zch}  
\end{center} \vspace{-6mm}
  \end{figure}
\normalsize
From (\ref{eqn1}),(\ref{up2}) and (\ref{c3}), we have
\footnotesize
\begin{equation}
D_s\leq \min\left(\frac{1}{2}(\max(M_1,N) +\max(M_2, N)-N_E), M-N_E,N\right)\\
\vspace{-1mm}
\end{equation} 
\normalsize
\end{proof}

\section{Achievable scheme}
\label{sec:scheme}
For securing the legitimate messages, the transmitters uses a two-step noise injection by simultaneously sending a jamming signal and using a stochastic encoder as follows,

\begin{enumerate}
\item The transmitters send a jamming signal with power $P^J=\alpha P$ that guarantees that all eavesdropper have a constant rate ($o(log P)$) for all legitimate signal power values,  where $\alpha$ is a constant controlled the transmitters to adjust the jamming.
\item A stochastic encoder is built using random binning. The encoder randomness is designed to be larger that any of the post-jamming eavesdroppers leakage, hence all eavesdroppers would have zero rate with the code length goes to infinity meeting the secrecy constraints in \eqref{eqn:cond}.
\end{enumerate} 
The jamming signal transmitted is a $N_E$ vector $\vec{r}=[\vec{r}_1 \text{ } \vec{r}_2]^T$ with random symbols using $\vec{V}^J_{1}$ and $\vec{V}^J_{2}$ as jamming precoders\footnote{For the special case $N_E=1$, only one user sends a single jamming symbol.}. Hence, the transmitted coded signal can be broken into legitimate signal, $\vec{s}_i$, and jamming signal, $\vec{r}_i$, such that 
$$\vec{x}_i = \left[\begin{array}{c} \vec{s}_i\\  \vec{r}_i \end{array}\right], i \in\{1,2\}.$$
Accordingly, the precoder, $\vec{V}_i$ can be also broken into legitimate precoder, $\vec{V}^L_i$, and jamming precoder, $\vec{V}^J_i$ such that
$$\vec{V}_i = \left[\begin{array}{cc} \vec{V}^L_i &  \vec{V}^J_i \end{array}\right] i \in\{1,2\}.$$
  
Choosing $\vec{V}^J$ to be the unitary matrix, the jamming power becomes $P^J=\text{E}\{\text{tr}(\vec{r}_i\vec{r}_i^{\dagger})\} = \alpha P$, where $\alpha$ is a constant controlled by the transmitter.

\begin{proposition}
The jamming signal, $\vec{r}$, overwhelms \textit{all} eavesdroppers' signal space, and all eavesdroppers end up decoding zero DoF of the legitimate messages. The transmitter then uses a stochastic encoder to satisfy the secrecy constraint in \eqref{eqn:cond}\\
\end{proposition}
Let $\bar{R}_e= I(\vec{Z}; \vec{s}_1, \vec{s}_2)$ be the rate of the eavesdropper with the best channel assuming it also has $N_E$ antennas. Let $R_e = I(\vec{Z}; \vec{W}_1, \vec{W}_2)$ be the legitimate message rate of the same eavesdropper, where $R_e < \bar{R}_e$ because of the stochastic encoder used. let $\bar{R}_ej$ be the rate of the $j$th eavesdropper. Then $\bar{R}_ej\leq \bar{R}_e \forall j \in L$, where $L$ is the unknown number of eavesdroppers.
\begin{proof}
\begin{eqnarray}
\nonumber n \bar{R}_e && \leq I(\vec{Z}^n; \vec{s}_1^n, \vec{s}_2^n)\\
\nonumber  && = h(\vec{Z}^n) -h(\vec{Z}^n|\vec{s}_1^n, \vec{s}_2^n)\\
\nonumber \bar{R}_e && \leq N_E log P -N_E log P^J + C\\
\nonumber     && \leq N_E log P -N_E log \alpha P +C\\
&&\leq C=o(log P)  \label{con}
\end{eqnarray}
\noindent where $C$ is a constant that does not depend on $P$ and known to the transmitter. 
\end{proof}
\begin{remark}
The constant eavesdropper rate comes from the fact that $P^J$ is controlled by the transmitter.  Hence, setting $P^J=\alpha P$, a constant SNR is guaranteed at the eavesdroppers and a constant rate independent of $P$. For the case of the constant known eavesdropper channel or unknown fading channel with known statistics, the constant $C$ is known transmitter.
\end{remark}

The transmitters use the rate difference to transmit perfectly secure messages using a stochastic encoder similar to the one described in~\cite{khan} according to the strongest eavesdropper's rate, $C$, in worst case scenario to achieve the secrecy constraint in \ref{eqn:cond} . Let 
\begin{equation}
C_i \in C(R^t_i, R_i, n) \forall i=1,2
\end{equation}

denote a Wyner code of size $2^{nR^t_i}$ to encode a confidential message set $W_i =
\{1, 2, . . . , 2^{nR_i} \}$ of transmitter $i$, where $R^t_i \geq R_i$ and $n$ is the codeword length. Therefore, there are two rates that define the Wyner code, the legitimate channel code rate $R^t_i$
and the secure message rate $R_i$. As a result, the rate $R^l=R^t_i - R_i$ is the cost of secrecy or the rate lost to secure the legitimate message. $R^l$ defines the amount of randomness added in a Wyner code. For a Wyner code, if $\hat{R}_e= R_l$, then the eavesdropper cannot decode the secure message sent (i.e $\lim_{n\longrightarrow \infty} \frac{1}{n} R_e \leq \epsilon)$.
The Wyner code $C(R^t_i, R_i, N)$ is built using random binning \cite{9}. We generate $2^{nR^t_i}$ codewords $s_i^n(w_i, v_i)$, where $w_i = 1, 2, . . ., 2^{nR_i} $,
and $v_i = 1, 2, . . ., 2^{n(R^t_i-R_i)}$, by choosing the $2^{nR_i^t}$ symbols $s_i(w_i, v_i)$ independently at random according to the input distribution $p(s_i)$. Then we distribute them randomly into  $2^{nR_i}$ bins such that each bin contains $2^{n(R_i^t-R_i)} $ codewords.
The stochastic encoder of $C(R_i^t, R_i, N)$ is described by a matrix of conditional
probabilities so that, given $w_i \in W_i$, we randomly and uniformly select a codeword to transmit  from the bin $w_i$ or in other words, we select $v_i$ from
$\{1, 2, . . . , 2^{n(R^t_i-R_i)}\}$ and transmit $s_i^n(w_i, v_i)$. We assume that the legitimate
receiver employs a typical-set decoder. Given the received signal $y^n$, the legitimate
receiver tries to find a pair $(\hat{w} , \hat{v})$ so that $s^n(\hat{w} , \hat{v})$ and $y^n$ are jointly typical \cite{9}. 
We set $R_i= I(\vec{s}_i, \vec{Y})- I(\vec{s}_i,\vec{Z}) -\epsilon$ and  $R_i^t=I(\vec{s}_i, \vec{Y})-\epsilon$. The error probability and equivocation calculations are
straight forward extensions of similar Wyner random binning
encoders (\cite{9}, \cite{bcs}).

\begin{eqnarray}
H(\vec{W_i}^n)\hspace{-2mm} &=& I(\vec{s}_i^n; \vec{Y}^n) - I(\vec{s}_i^n;\vec{Z}^n) -m\epsilon\\
H(\vec{W_i}^n|\vec{Z}^n)&=&I(\vec{s}_i^n; \vec{Y}^n|\vec{Z}^n) - I(\vec{s}_i^n;\vec{Z}^n|\vec{Z}^n)-n\epsilon\hspace{4mm}\\
&=& I(\vec{s}_i^n; \vec{Y}^n, \vec{Z}^n) - I(\vec{s}_i^n,\vec{Z}^n)-n\epsilon\\
&\geq& H(\vec{W_i}^n) -n\epsilon\\
\end{eqnarray}
and,
\begin{eqnarray}
H(\vec{W_1}^n,\vec{W_2}^n|\vec{Z}^n)\hspace{-2mm} &=& H(\vec{W_1}^n|\vec{Z}^n)+H(\vec{W_2}^n|\vec{Z}^n)\hspace{2mm}\\
&\geq& H(\vec{W_1}^n) + H(\vec{W_2}^n) -2n\epsilon\\
&\geq& H(\vec{W_1}^n, \vec{W_2}^n) -2n\epsilon
\end{eqnarray}

 \begin{eqnarray} \label{secrate}
\nonumber \sum_{i=1}^2 R_i \hspace{-7mm}&& \geq \frac{1}{2} \log \left| \vec{I}+ \sum_{i=1}^2 (\vec{U}\vec{H}_i\vec{V}_i^L\vec{s}_i\vec{s}^{\dagger}_i\vec{V}^{L\dagger}_i\vec{H}_i^{\dagger}\vec{U}^{\dagger})\right| -R_e\\ 
\end{eqnarray}

As $\lim_{n \longrightarrow \infty} \frac{1}{n} R_e \leq \epsilon$ for all values of $\vec{G}_i$ and $P$, a positive secrecy rate, which is monotonically increasing with $P$, is achieved. Computing the secrecy degrees of freedom boils down to calculating the degrees of freedom for the first term in the right hand side of \eqref{secrate}, which represents the receiver DoF after jamming is applied.
 
With the eavesdroppers completely blocked, it remains to show how the jamming signal directions are designed to achieve the maximum possible secure DoF. First, we study the secure DoF for $M \leq N$, then go for $M > N$ with different regions of the relations between $(M_1, M_2, N, N_E)$.

\subsection{Achievability for $M \leq N$}

For this region, transmitters one and two send the jamming signals using precoders $\vec{V}^J_{1}$ and $\vec{V}^J_{2}$, with dimensions $J_1$ and $J_2$, respectively,  such that $J_1+J_2=N_E$. 

\subsubsection*{Random jamming}
The jamming precoders and symbols are randomly chosen. We call this method random jamming. The receiver zero-forces the jamming signal using the post-processing matrix $\vec{U}$ as in~\eqref{zero}. Accordingly, $M-N_E$ secure DoF can be sent. Since $N \geq M $, the receiver can decode $M-N_E$ DoF after zero-forcing the jamming signal.\\

\begin{equation}\label{zero}
\vec{U}= [\vec{I} - \vec{a}\vec{a}^{-1}]
\end{equation}
 where
\begin{equation}
\vec{a}= \vec{H}_1 \vec{V}_1^J+\vec{H}_2 \vec{V}_2^J
\end{equation}

\subsection{Achievability for $M>N$}
For this region we use three methods for jamming, aligned jamming, nullspace jamming and random jamming. As random jamming was described above, the other two will be explained in the following.

\subsubsection*{Aligned jamming}
The jamming signals of both transmitters are aligned at the legitimate receiver signal space. 
Let $\mathcal{I}$ be the jamming space at the receiver. Each transmitter aligns a part or the whole of its jamming signal into this jamming space. The total signal space of transmitter one and transmitter two occupies  \emph{only} $M_1$ and $M_2$ dimensions, receptively, at the receiver. These two spaces are distinct if $M_1<N$, so a common space is needed to direct the jamming signal into.\\
\begin{figure}
  \begin{center}
\hspace{-4mm} \includegraphics[width=.50\textwidth]{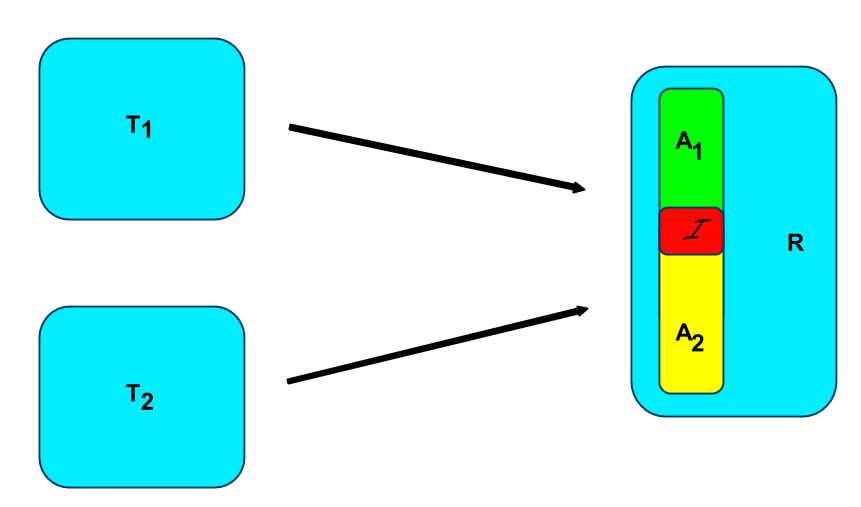}\vspace{-2mm}
\caption{The intersection of signal spaces at the receiver}
  \label{sys}
  \end{center} \vspace{-6mm}
  \end{figure}
    \vspace{1.5mm}
    
Let $A_1$ and $A_2$ span the received signal spaces of transmitter one and two at the receiver,
$\mathcal{I}$ is chosen to be the intersection of these two spaces, i.e.,
\begin{equation}\label{first}
\mathcal{I}= A_1\bigcap A_2.
\end{equation}
$\mathcal{I}$ would have positive size only if $M\geq N$~\cite{ours}.
Without loss of generality, we design $\vec{V}_1^J$ and $\vec{V}_2^J$ such that,
\begin{eqnarray}
\vec{H}_1\vec{V}^J_1= \vec{H}_2\vec{V}^J_2= \mathcal{I}\label{two}
\label{eqn:30}
\end{eqnarray}
\\
While the system of equations in (\eqref{two}) has more variables than the number of equations,~\eqref{first} ensures that the system has a unique solution as $\mathcal{I}$ lies in the spans of $\vec{H}_1$ and $\vec{H}_2$.

Let
\begin{equation}
 \vec{H}_i=
\begin{bmatrix}
 \vec{H}_i^{'}\\
 \vec{H}_i^{''}
\end{bmatrix}
\;\; \forall \; i=1,2,
\end{equation}
where $\vec{H}_i'$ contains the  $M_i$ rows of $\vec{H}_i$ and $\vec{H}_i''$ contains the other $N-M_i$ rows.

Let
\begin{equation}
 \mathcal{I}=
\begin{bmatrix}
 \mathcal{I}_i' \\
\mathcal{I}_i''
\end{bmatrix}
\;\; \forall \; i=1,2,
\end{equation}

where $\mathcal{I}_i'$ contains the  $M_i$ rows of $\mathcal{I}$ and $\mathcal{I}_i''$ contains the other $N-M_i$ rows.

Therefore, we can choose the following design which satisfies ~\eqref{eqn:30}
\begin{equation}\label{last}
\vec{V}^J_1 =(\vec{H}_1')^{-1}\mathcal{I}_1'
\end{equation}
\begin{equation}
\vec{V}^J_2 =(\vec{H}_2')^{-1}\mathcal{I}_2'
\end{equation}
\\
For the legitimate receiver to remove the jamming signal and decode the legitimate message, it zero forces the jamming signal using the post-processing matrix $\vec{U}$ as in (\ref{zero}).\\

For the case $N_E$ is odd, each transmitter will align its jamming signal into an $\lfloor \frac{N_E}{2} \rfloor$--dimensional half space using linear alignment.  The remaining $1$dimensional space will be equally shared between the two transmitters' jamming signal using real interference alignment~\cite{sennur_helpers}, yielding each transmitter's jamming signal to occupy $\frac{N_E}{2}$ .

\subsubsection*{Nullspace jamming}
In nullspace jamming method, transmitter one sends a jamming signal of $J_1$ dimensions using the precoder $\vec{V}^J_{1}$ which lies in the nullspace of the channel $\vec{H}_1$, while transmitter two sends a jamming signal of $J_2$ dimensions using the precoder $\vec{V}^J_2$ which lies in the nullspace of the channel $\vec{H}_2$.
\vspace{-2mm}
\begin{eqnarray}
\vec{V}^J_i=\text{\textit{Null}} (\vec{H}_i) i\in 1,2
\end{eqnarray}
\noindent This blocks $J_1+J_2$ dimensions at the eavesdropper and leaves $N$ free dimensions at the legitimate receiver to attain the legitimate signal in addition to other jamming if needed.\\

\noindent\emph{Case $M_1<N \hspace{1mm} \text{and} \hspace{1mm} N_E \geq 2(M-N)$:}

In this region, aligned jamming and random jamming are used. The first transmitter signal is divided into two parts of sizes $J_1$ and $J_2$ , The first part uses aligned jamming precoder with $J_1=\text{min}({\frac{N_E}{2}, M-N})$, or $J_1=(M-N)$, while the second part uses random jamming precoder with $J_2=[N_E-2J_1]^+$. The second transmitter jamming signal size is $J_1$ and uses aligned jamming precoder. This scheme wastes $2J_1+J_2=N_E$ dimensions of the transmitters signal space for jamming, so they can transmit at $M-N_E$ DoF to the receiver. On the other hand, the jamming occupies $J_s=J_1+J_2$ dimensions at the receiver. 

\begin{lemma} For the scheme proposed in the case $M_1<N \hspace{1mm} \text{and} \hspace{1mm} N_E \geq 2(M-N)$, a sufficient condition for the receiver to decode the transmitted $M-N_E$ DoF is given by $N-J_s \geq M-N_E$.
\end{lemma}
\vspace{-2mm}
\begin{proof}
\begin{eqnarray}
\nonumber J_s&=& J_1+J_2\\
\nonumber J_s&=& N_E-J_1\\
\nonumber J_s&=& N_E-M+N\\
\nonumber N-J_s&=& N-N_E+M-N\\
\nonumber N-J_s&=&M-N_E.
\end{eqnarray}
Thus, $M-N_E$ is achievable.
\end{proof}

\noindent\emph{Case $M_1<N \text{ and } N_E < 2(M-N)$:}

In this region, aligned jamming is used alone with both transmitters jamming signals sizes set to $J_1=\frac{N_E}{2}$. Consequently, the receiver loses $J_s=\frac{N_E}{2}$ dimensions because of the jamming.

The achievable SDoF for this region is
\begin{equation}
\nonumber d_1+d_2=\text{min}\bigg(N-\frac{N_E}{2}, M-N_E\bigg)\\
\end{equation}
Considering that in this case $M_1<N \text{ and } N_E < 2(M-N)$, we get
\begin{equation}
\nonumber d_1+d_2=N-\frac{N_E}{2},
\end{equation}
which can be rewritten as
\begin{equation}
d_1+d_2=\frac{\text{max}(M_1,N)+\text{max}(M_2,N)-N_E}{2}\\
\end{equation}

\noindent\emph{Case $M_1>N \hspace{1mm} \text{ and } \hspace{1mm}N_E< M_1-N+[M_2-N]^+ $:}

In this region, nullspace jamming is used alone. Transmitter one sends $J_1=\text{min}(N_E, M_1-N)$ dimensional jamming signal and transmitter two sends $J_2=N_E-J_1$ dimensional jamming signal in the null spaces of the legitimate receiver channels.
This leaves the receiver with $N$ jamming free dimensions to decode the $N$ SDoF transmitted.  Consequently, the upperbound $N$ is achieved.\\

\noindent\emph{Case $M_1>N, M_2<N \text{ and } N_E \geq [M_1-N+2M_2] $:}

In this region, all three jamming methods are used, the first transmitter jamming signal is divided into three parts of sizes $J_1=[M_1-N]^+$, $J_2=\min(M_2, \frac{N_E-J_1}{2})$, which is equivalent to $J_2= M_2$, and $J_3=N_E-(J_1+J_2)$. While the second transmitter jamming signal size is $J_2$. The first transmitter uses nullspace jamming, aligned jamming and random jamming for its three parts, respectively. The second transmitter uses aligned jamming only. This scheme uses $N_E$ dimensions of the transmitters signal space, so they can transmit at $M-N_E$ DoF to the receiver. On the other hand, the jamming occupies $J_s=J_2+J_3$ dimensions at the receiver. 
\begin{lemma}
Using the proposed achievable scheme under the case $M_1>N, M_2<N \text{ and } N_E \geq [M_1-N+2M_2] $, the condition $N-J_s\geq M-N_E$ is achieved, which implies that the legitimate signal and the jamming spaces are not overlapping.
\end{lemma}
\vspace{-7mm}
\begin{proof}
\begin{eqnarray}
\nonumber J_2 &=&M_2\\
\nonumber J_3 &=&N_E-[(M_1-N)+2M_2]\\
\nonumber J_s &=&N_E-[(M_1-N)+M_2]\\
\nonumber N-J_s &=&N-N_E+M_1-N+M_2\\
\nonumber N-J_s &=&M-N_E
\end{eqnarray}
\end{proof}
\vspace{-5mm}
\noindent\emph{Case $M_1>N,  M_2<N \text{ and } M_1-N\leq N_E< (M_1-N+2M_2)$:}

In this region, two jamming methods are used, the first transmitter jamming signal is divided into two parts of sizes $J_1=[M_1-N]^+$, $J_2=\text{min}(M_2, \frac{N_E-J_1}{2})$, while the second user jamming signal size is $J_2$. The first transmitter uses nullspace jamming for its first part and aligned jamming for the second part, which is aligned to the second transmitter's jamming signal at the receiver. The jamming occupies $J_s=J_2$ dimensions at the receiver. 
\begin{lemma}
For the proposed scheme under the case $M_1>N,  M_2<N \text{ and } M_1-N\leq N_E< (M_1-N+2M_2)$, the achievable secure degrees of freedom is given by 
\begin{equation}
d_1+d_2\leq\frac{\text{max}(M_1,N)+\text{max}(M_2,N)-N_E}{2}
\end{equation}
\end{lemma}
\begin{proof}
\begin{eqnarray}
\nonumber J_s&=&  \frac{N_E-[M_1-N]}{2}\\
\nonumber N-J_s&=&  N-\frac{N_E-[M_1-N]}{2}\\
\nonumber N-J_s&=& \frac{2N-N_E+[M_1-N]}{2}\\
\nonumber N-J_s&=& \frac{N-N_E+M_1}{2} 
\end{eqnarray}

Thus, 
\begin{eqnarray}
\nonumber d_1+d_2 \leq \text{min}(M-N_E, \frac{N-N_E+M_1}{2})\\ 
\nonumber d_1+d_2 \leq  \frac{N-N_E+M_1}{2},
\end{eqnarray}
which can be rewritten as, 
\begin{equation}
\nonumber d_1+d_2 \leq \frac{\text{max}(M_1,N)+\text{max}(M_2,N)-N_E}{2}.
\end{equation}
\end{proof}

\noindent\emph{Case $M_1>N,  M_2\geq N \text{ and } N_E \geq M-2N$:}
\vspace{3mm}\\
In this region, two jamming methods are used. Each transmitter jamming signal is divided into two parts of sizes 
\begin{eqnarray}
\nonumber J_{1,i}=M_i-N, \;\;\;\;\ \forall \; i\in\{1,2\}\\
\nonumber J_2= \frac{N_E-\sum_{i\in\{1,2\}}J_{1,i}}{2}.
\end{eqnarray}
Both transmitters use nullspace jamming for their first part and aligned jamming for the second part. The jamming occupies $J_s=J_2$ dimensions at the receiver.

\begin{lemma}
For the proposed scheme under case $M_1>N,  M_2\geq N \text{ and } N_E \geq M-2N$, the achievable secure degrees of freedom is
\begin{equation}
d_1+d_2 \leq \frac{\text{max}(M_1,N)+\text{max}(M_2,N)-N_E}{2}
\end{equation}
\end{lemma}

\begin{proof}
\begin{eqnarray}
\nonumber J_s &=& \frac{N_E-(M_1-N+M_2-N)}{2}\\
\nonumber N-J_s &=&  N- \frac{N_E-[M-2N]}{2}\\
\nonumber N-J_s &=&\frac{2N-N_E+M-2N}{2}\\
\nonumber N-J_s &=&   \frac{M-N_E}{2}.
\end{eqnarray}
Thus, 
\begin{eqnarray}
d_1+d_2 \leq  \frac{M-N_E}{2},
\end{eqnarray}
which can be rewritten as
\begin{equation}
\nonumber d_1+d_2 \leq \frac{\text{max}(M_1,N)+\text{max}(M_2,N)-N_E}{2}
\end{equation}
\end{proof}

\section{Conclusion}
\label{sec:conclusion}
We studied the two-transmitter Gaussian multiple access wiretap channel with multiple antennas at the transmitters, legitimate receivers and eavesdroppers. Generalizing new upperbound was established and  a new achievable scheme was provided. We used the new optimal scheme to derive the sum secure DoF of the channel. We showed that the our scheme meets the upperbound or all $M_1, M_2, N_E$ combinations. We showed that Cooperative Jamming is SDoF optimal even without the eavesdropper CSI available at the transmitters by showing that jamming signal independent of the eavesdropper channel and only depends on the signal transmitted power make the eavesdropper decoded DoF. Finally we showed that if any eavesdropper has more antennas that the sum of the transmitting antennas or the receiving antennas the SDoF is zero.


\begin{thebibliography}{1}
    
    \bibitem{wyner} A. D. Wyner. {\em The wiretap channel}.   Bell systems technical journal, vol. 8, Oct. 1975.
		
		    \bibitem{csi} I. Csiszar and J. Korner.  {\em  Broadcast channels with confidential messages}.   IEEE transactions on information theory, vol.24, no. 3, pp:339--348, May 1978.
						
		    \bibitem{leu} S. K. Leung-Yan-Cheong and M. E. Hellman.  {\em  Gaussian wiretap channel}. IEEE transactions on information theory, vol.24, no. 4, pp:451--456, July 1978.
	
	
  \bibitem{secop} E. Ekrem and S. Ulukus.  {\em Secrecy in cooperative relay broadcast channels}. IEEE International Symposium on Information Theory, 2008. ISIT 2008.  



  \bibitem{mimo_wire} A. Khisti, G. W. Wornell, A. Wiesel, and Y. Eldar. {\em On the Gaussian
MIMO wiretap channel}. IEEE International Symposium Information Theory Proceedings, Jun. 2007, pp. 2471--2475.
    
      \bibitem{mimo_secure} A. Khisti and G. Wornell. {\em Secure transmission with multiple antennas
I: The MISOME wiretap channel }.  IEEE Transactions on Information Theory,
vol. 56, no. 7, pp. 3088--3104, Jul. 2010.

         \bibitem{mimo_note} T. Liu and S. S. Shamai. {\em Note on the secrecy capacity of the multiple-
antenna wiretap channel}. IEEE Transactions on Information Theory, vol. 55, no. 6, pp. 2547--2553, Jun. 2009.
        
         \bibitem{mimo_hass} F. Oggier and B. Hassibi. {\em The secrecy capacity of the MIMO wiretap
channel}.  IEEE Trans. Inf. Theory, vol. 57, no. 8, pp. 4961--4972, Aug. 2011.

        \bibitem{yener_coop}  E. Tekin and A. Yener. {\em Achievable rates for the general Gaussian
multiple access wire-tap channel with collective secrecy }. In 44th Annual Allerton Conference on Communication, Control and Computing, Septtember 2006.

\bibitem{yener_mimo}  X. He and A. Yener. {\em MIMO Wiretap Channels with Arbitrarily Varying Eavesdropper Channel States}. Arxiv.org:1007.4801.



\bibitem{yener_mac} X. He, A. Khisti and A. Yener. {\em MIMO Multiple Access Channel With
an Arbitrarily Varying Eavesdropper: Secrecy Degrees of Freedom}. IEEE Transactions on Information Theory, Vol. 59, No. 8, pp.
4733--4745, Aug. 2013.
        
   \bibitem{mimo_shafie}  S. Shafiee, N. Liu, and S. Ulukus. {\em Towards the secrecy capacity of the
Gaussian MIMO wire-tap channel The 2-2-1 channel }.  IEEE Transactions on
Information Theory, vol. 55, no. 9, pp. 4033--4039, Sept. 2009.

\bibitem{mimo_confidential}  R. Liu, T. Liu, and H. V. Poor. {\em multiple-input multiple-output
Gaussian broadcast channels with confidential messages}. IEEE
Trans. Inf. Theory, vol. 56, no. 9, pp. 4215--227, Sept. 2010.

\bibitem{sennur_mimo}  E. Ekrem and S. Ulukus. {\em The secrecy capacity region of the Gaussian
MIMO multi-receiver wiretap channel}. IEEE Trans. Inf. Theory, vol.
57, no. 4, pp. 2083--114, Apr. 2011.

 

\bibitem{sennur_mac} Jianwei Xie and S. Ulukus. {\em 
Secure Degrees of Freedom of the Gaussian Multiple Access Wiretap Channel}. IEEE International Symposium on  Information Theory Proceedings (ISIT), Jul. 2013. 

\bibitem{sennur_helpers} Jianwei Xie and Sennur Ulukus {\em 
Secure degrees of freedom of the Gaussian wiretap channel with helpers and no eavesdropper CSI: Blind cooperative jamming
 } The Annual Conference on
 Information Sciences and Systems (CISS), March 2013. 

\bibitem{khan} Ghadamali Bagherikaram, Abolfazl S. Motahari, Amir K. Khandani
{\em On the Secure Degrees-of-Freedom of the Multiple-Access-Channel} IEEE Transaction Information
Theory, submitted March 2010. Also available at [arXiv:1003.0729]
\bibitem{9} Rouheng liu, Wade Trappe {\em Securing Wireless communications at the physical layer} Springer US, 2010.
\bibitem{bcs} Ghadamali Bagherikaram, Abolfazl S. Motahari, Amir K. Khandani {\em The Secrecy Rate Region of the Broadcast Channel}, arxiv.org: 0806.4200.

Ghadamali Bagherikaram, Abolfazl S. Motahari, Amir K. Khandani
\bibitem{other} Mohamed Amir, Tamer Khattab, Tarek Elfouly, Amr Mohamed {\em The Secure Degrees of Freedom of the MIMO Multiple Access Channel with Multiple Unknown
Eavesdroppers}, DOI: 10.13140/RG.2.1.2769.5440.
\bibitem{ours}
M. Amir, A. El-Keyi, and M. Nafie,
{\em Constrained interference alignment and the spatial degrees of freedom of MIMO cognitive networks}
IEEE Transactions on Information Theory, vol. 57., no. 5, pp 2994-3004
%\bibitem{sennur_uni} Jianwei Xie and Sennur Ulukus {\em 
%Uni?ed Secure DoF Analysis of K-User
%Gaussian Interference Channels} IEEE International Symposium on
 %Information Theory Proceedings (ISIT), July 2013. 
  \end{thebibliography}
\end{document}